\documentclass[11pt]{article}

\usepackage{latexsym, amscd, amsfonts, eucal, mathrsfs, amsmath, amssymb, amsthm, stmaryrd, fullpage}
\usepackage{tikz}
\usetikzlibrary{calc,positioning,fit}
\usetikzlibrary{fit,svg.path,positioning}

\def\la{\langle}
\def\ra{\rangle}
\def\p{\prime}

\def\bZ{\mathbb Z}

\def\bN{\mathbb N}

\def\sse{\subseteq}

\def\lm{\lambda}

\def\sm{\sigma}

\def\ot{\otimes}

\def\id{\mathrm{id}}

\def\bap{\bigcap}
\def\bup{\bigcup}

\def\ALL{\mathrm{ALL}}

\def\CONS{\mathrm{CONS}}
\def\SWAP{\mathrm{SWAP}}
\def\hyp{\textrm{-}}
\def\lcm{\mathrm{lcm}}
\def\CC{\mathrm{CC}}

\newtheorem{thm}{Theorem}[section]
\newtheorem{lemma}[thm]{Lemma}
\newtheorem{prop}[thm]{Proposition}
\newtheorem{coro}[thm]{Corollary}

\theoremstyle{definition}

\newtheorem{defn}[thm]{Definition}
\newtheorem{eg}[thm]{Example}

\newtheorem{notn}[thm]{Notation}

\newtheorem{rmk}[thm]{Remark}

\begin{document}
\title{Some Results on Reversible Gate Classes Over Non-Binary Alphabets}

\author{Yuzhou Gu\thanks{MIT. \ Email: yuzhougu@mit.edu.}}
\date{}
\maketitle

\begin{abstract}
We present a collection of results concerning the structure of reversible gate classes over non-binary alphabets, including (1) a reversible gate class over non-binary alphabets that is not finitely generated (2) an explicit set of generators for the class of all gates, the class of all conservative gates, and a class of generalizations of the two (3) an embedding of the poset of reversible gate classes over an alphabet of size $k$ into that of an alphabet of size $k+1$ (4) a classification of gate classes containing the class of $(k-1,1)$-conservative gates, meaning gates that preserve the number of occurrences of a certain element in the alphabet. 
\end{abstract}
\tableofcontents
%\IEEEpeerreviewmaketitle

  \section{Introduction}
  The ``pervasiveness of universality'', described in Aaronson et al.\ \cite{AGS15}, is the phenomenon that in a class of operations, a small number of simple operations is likely to generate all operations.
  This phenomenon is central in the theory of computation. 
  Yet in some cases, systems fail to be universal, and different ways of failure reveal a rich structure.
  An example is Post's lattice \cite{Pos41}, one of the most remarkable results in the early age of computer science, which is a complete description of all ways in which a set of classical gates over the binary alphabet can fail to be universal. 
  
  Inspired by Post's lattice, Aaronson et al.\ \cite{AGS15} proposed the problem of classifying all quantum gate classes. However, quantum gate classes turn out to be much more complicated than classical gates classes.
  Therefore, instead of studying quantum gates directly, they studied reversible gates, which are in some sense the correct classical analogue of quantum gates. 
  In op. cit., Aaronson et al.\ gave a complete classification of reversible gate classes over the binary alphabet. 
  They also pointed out several directions for further research. 
  One direction is to classify stabilizer operations over qubits, which has been completed recently by Grier and Schaeffer \cite{GS16}.
  The current paper, taking another direction, studies the behavior of reversible gate classes when the alphabet is not binary.
  
  \subsection{Our results}
  Our results include the following:
  \begin{enumerate}
  \item A reversible gate class over non-binary alphabets that is not generated by a finite subset of reversible gates (Theorem \ref{NonFinGenClass}). This answers a question of Aaronson et al.\ \cite{AGS15}.
  \item An explicit set of generators for the class of all gates, the class of all conservative gates (Theorem \ref{AllGenerate} and Theorem \ref{ConsGen}). We define generalizations of these two classes (Definition \ref{LmCons}), and give an explicit set of generators for them (Theorem \ref{LmConsGeneration}).
  \item An embedding of the poset of reversible gate classes over an alphabet of size $k$ into the poset of reversible gate classes over an alphabet of size $k+1$ (Theorem  \ref{PosetEmbedding}). 
  \item A classification of reversible gate classes containing the class of reversible $(k-1,1)$-conservative gates (Theorem \ref{CONSInteger}). 
  \end{enumerate}
  These results show that there are differences and similarities between reversible gate classes over the binary alphabet and those over a non-binary alphabet.
  It is known that 1) and 4) do not hold in the binary case. 
  Yet, 2) and 4) can be seen as a generalization of the corresponding results in the binary case.
  
  \subsection{Related work}
  %There has not been much previous work on reversible gate classes over non-binary alphabets.
  Aaronson et al.\ \cite{AGS15} completely classified reversible gate classes over the binary alphabet. 
  Many of their methods for proving generation results apply to non-binary case, e.g. \cite{AGS15} Theorem 23, Theorem 25, Theorem 41.
  Je{\v r}{\'a}bek \cite{Jer14} gave a reversible clone-coclone duality (called ``master clone-coclone duality'' in that work), generalizing the classical clone-coclone duality. Je{\v r}{\'a}bek's result works for non-binary alphabets.
  
  In this paper reversible gate classes are closed under the ancilla rule by definition (Definition \ref{GenRules}). In some previous work the ancilla rule is removed. 
  Boykett \cite{Boy15} gave a finite set of generators with no ancillas for the class of all reversible gates when the alphabet size is odd. Boykett's result implies Theorem \ref{AllGenerate} when $|A|$ is odd. 
  Boykett et al.\ \cite{BKS16} gave several non-finitely generated gate classes when borrowed symbols (called ``borrowed bits'' in the paper) are allowed but ancilla symbols are not allowed. 
  Because borrowed symbols are weaker than ancilla symbols, Boykett et al.'s result does not give a non-finitely generated gate class when ancilla symbols are allowed. 
  
  Aaronson et al.\ \cite{AGS15} and Xu \cite{Xu15} studied the number of ancilla bits (or borrowed bits) used in generation over the binary alphabet. In this paper we do not try to minimize the number of ancilla symbols used. 
  
%   Classical gate classes have been studied much earlier than reversible gate classes. Post \cite{Pos41} completely classified classical gate classes over the binary alphabet. 
%   Since then, there has been intensive study on classical gate classes over non-binary alphabets, under the name ``function algebras''. Lau \cite{Lau06} is a monograph on this subject. Results on classical gate classes over non-binary alphabets often hint the behavior of reversible gate classes over non-binary alphabets.
  
  \section{Acknowledgments}
This work is supported in part by MIT SuperUROP Undergraduate Research Program and MIT Lincoln Laboratory.
The author wishes to express his thanks to Prof. Scott Aaronson for proposing this interesting topic and giving useful advice, and to Luke Schaeffer for helpful mentoring.
The author also thanks Cameron Musco, Dr. Kevin Obenland and Prof. Yury Polyanskiy for helpful discussions and Thalia Rubio for writing advice.
  
  \section{Preliminaries}
  In this section we give basic definitions of the object of study.
  
  Fix a finite set $A$ with $|A|\ge 2$. $A$ is called the \textbf{alphabet}.
  A theory of reversible gates can be defined for arbitrary $A$, including infinite $A$. However, we restrict our attention to the case $A$ is finite in this paper. 
  
  By a \textbf{classical gate} we mean a function of sets $A^k\to A$ for some $k\in \bN$. (We take $\bN=\bZ_{\ge 0}$.) In contrast, a reversible gate is a reversible function, which then requires the domain and the codomain to have the same cardinality.
  \begin{defn}
  A \textbf{reversible gate} is a bijective function $A^k\to A^k$ for some $k\in \bN$.
  \end{defn}

  \begin{notn}\label{GateDefn}
  We give some examples of gates and build notations that will be used in later sections. 
  
  Let $u,v\in A^k$ be two strings. Define reversible gate $\tau_{u,v}:A^k\to A^k$ to be the function that maps $u$ to $v$, $v$ to $u$, and all other inputs to themselves.
  
  Let $u\in A^k$, $v\in A^l$ be two strings. Define reversible gate $\SWAP_{u,v}:A^{k+l}\to A^{k+l}$ to be the function that maps $uv$ to $vu$, $vu$ to $uv$, and all other inputs to themselves. 
  In other words, $\SWAP_{u,v}=\tau_{uv,vu}$.
  
  Let $F:A^k\to A^k$ be a gate and $w\in A^l$ be a string. 
  Define reversible gate $w\hyp F:A^{k+l}\to A^{k+l}$ that maps $wu$ to $wF(u)$ where $u\in A^k$, and maps all other inputs to themselves.
  \end{notn}
  
  \begin{rmk}
  Let $A=\{0,1\}$. Then $11\hyp \tau_{0,1}$ is the Toffoli gate.
  $1\hyp \SWAP_{0,1}$ is the Fredkin gate. 
  \end{rmk}

  In the theory of classical gates, there is a notion of generation of gates, which refers to the process of creating new gates from existing gates. Similarly, in the reversible case, we also have generation of gates. 
  
  \begin{defn}\label{GenRules}
  There are several ways in which reversible gates can be generated.
  \begin{enumerate}
  \item \textbf{Permutation rule.} Let $k$ be a non-negative integer and $\sm$ be a permutation of $\{1,\ldots,k\}$. We define the permutation gate $P_\sm:A^k\to A^k$ that maps $(a_1,\ldots,a_k)\in A^k$ to $(a_{\sm(1)},\ldots,a_{\sm(k)})\in A^k$. Permutations of symbols come for free: we can generate $P_\sm$ from nothing. 
  
  \item \textbf{Tensor product rule.} Assume we have two reversible gates $F:A^k\to A^k$, $G:A^l\to A^l$. Then we can generate their tensor product $F\ot G:A^{k+l}\to A^{k+l}$ which sends $(a_1,\ldots,a_k,b_1,\ldots,b_l)\in A^{k+l}$ to $(F(a_1,\ldots,a_k),G(b_1,\ldots,b_l))$. 
  
  \item \textbf{Composition rule.} Assume that we have two reversible gates  $F:A^k\to A^k$, $G:A^k\to A^k$. Then we can generate composition $F\circ G:A^k\to A^k$ which sends $(a_1,\ldots,a_k)\in A^k$ to $F(G(a_1,\ldots,a_k))$.
  
  \item \textbf{Ancilla rule.} Assume that we have a reversible gate $F:A^k\to A^k$. Assume there exists $0\le l\le k$, $(a_1,\ldots,a_l)\in A^l$ and a reversible gate $F^\p:A^{k-l}\to A^{k-l}$ such that for any $(b_1,\ldots,b_{k-l})\in A^l$, $F$ maps $(a_1,\ldots,a_l,b_1,\ldots,b_{k-l})$ to $(a_1,\ldots,a_l,F^\p(b_1,\ldots,b_{k-l}))$. 
  Then $F^\p$ can be generated from $F$. 
  \end{enumerate}
  \end{defn}

  \begin{rmk}
  Generation of reversible gates can be understood as the following process. 
  \begin{enumerate}
  \item Start with $n$ symbols. These are the inputs and we do not have control over their initial values. 
  \item Introduce $m$ ancilla symbols. The initial values of these symbols are decided by us.
  \item Apply some sequence of reversible gates. An application of a reversible gate $G$ with $k\le n+m$ inputs is first choosing $k$ different symbols among the $n+m$ ones, and then applying $G$ on the $k$ chosen symbols. 
  \item After step (3) finishes, make sure that the values of the $m$ ancilla symbols are the same as their initial values. (However, their values can change during step (3).) Remove the $m$ ancilla symbols. 
  \item  The set of all configurations of the $n$ input symbols is $A^n$. 
  The above process gives a function from the set of initial configurations to the set of final configurations, written as $F:A^n\to A^n$. 
  $F$ is actually a reversible gate, and is the reversible gate generated by this process.
  \end{enumerate}
  
  This way of understanding generation of reversible gates is very helpful and is used throughout this paper. 
  \end{rmk}
  
  \begin{rmk}
  The first three generation rules are natural and easy to understand. For a detailed discussion of the ancilla rule, see \cite{AGS15}, Section 1.2. 
  
  The set of generation rules given in Definition \ref{GenRules} is equivalent to the one in \cite{AGS15}, Section 2.2. 
  \end{rmk}
  
  With the generation rules, we can define gate classes. 
  \begin{defn}
  Let $S$ be a set of reversible gates. 
  We say $S$ is a \textbf{reversible gate class} if $S$ is closed under the generation rules in Definition \ref{GenRules}.
  \end{defn}
  \begin{defn}
  Let $S$ be a set of reversible gates. Define $\la S\ra$ to be the smallest reversible gate class containing $S$. We say $\la S\ra$ is the \textbf{reversible gate class generated by} $S$.
  Let $F$ be a reversible gate. We say \textbf{$S$ generates $F$} if $F\in \la S\ra$. 
  \end{defn}
  
  \begin{rmk}
  $\la S\ra$ is equal to the intersection of all reversible gate classes containing $S$.
  \end{rmk}
  
  \begin{eg}
  We give some examples of reversible gate classes. 
  
  Let $\ALL$ be the set of all reversible gates. Then $\ALL$ is a reversible gate class. When $|A|=2$, $\ALL$ is the reversible gate class generated by the Toffoli gate.
  
  Let $\CONS$ be the set of all conservative gates, i.e. reversible gates $F$ such that for every $a\in A$ and every input to $F$, the number of $a$'s in the input of $F$ is the same as the number of $a$'s in the output. Then $\CONS$ is a reversible gate class.
  When $|A|=2$, $\CONS$ is the reversible gate class generated by the Fredkin gate. 
  \end{eg}
  
  The set of reversible gate classes has a natural structure of a poset. 
  \begin{defn}
  We give the set of all reversible gate classes a partial order $\le$. For two reversible gate classes $S$, $T$, we have $S\le T$ if and only if $S\sse T$ as sets. 
  \end{defn}
  The structure is actually more than a poset, as shown in the following proposition. 
  \begin{prop}
  Let $L$ be the poset of reversible gate classes. 
  Then $L$ is a complete lattice. 
  \end{prop}
  \begin{proof}
  Let $\{S_i:i\in I\}$ be a set of reversible gate classes where $I$ can be infinite.
  Then the join of $\{S_i\}$ is $\la \bup_{i\in I} S_i\ra$ and the meet of $\{S_i\}$ is $\bap_{i\in I} S_i$.
  \end{proof}
  \begin{notn}
  In the following sections, when there is no ambiguity, we sometimes say ``gate'' instead of ``reversible gate'', and say ``gate class'' instead of reversible gate class.
  \end{notn}

  \section{Non-finite generation}
  When $|A|=2$, it follows from Aaronson et al.'s classification that every gate class is generated by a single gate (\cite{AGS15}, Corollary 5). 
  In this section we show that this fails when $|A|>2$. Actually, we prove that there are gate classes that are not generated by a finite set of gates.

  \begin{thm}\label{NonFinGenClass}
  If $|A|\ge 3$, there exists a gate class that is not finitely generated.
  \end{thm}
  
  \begin{proof}
  Assume without loss of generality that $A=\{1,\ldots,|A|\}$. 
  For $k\ge 1$, define $T_k=\tau_{1^k2,1^k3}$.
  Clearly $T_{k+1}$ generates $T_k$ by the ancilla rule. 
  Let $T=\la T_k:k\ge 0\ra$. We claim that $T$ is not finitely generated. 
  
  Assume for the sake of contradiction that $T$ is finitely generated. Then $T$ is generated by some finite set $I\sse T$. 
  Let $F\in I$ be a gate. Then $F\in I\sse T=\la T_k:k\ge 0\ra$. So $F$ is generated by a finite subset of gates in the form $T_k$. $T_{k+1}$ generates $T_k$, so $F$ is generated by a single gate $T_n$ for some $n$. 
  Because $I$ is finite, we can take $n$ to be large enough, so that all gates in $I$ are generated by $T_n$. 
  We know that $T$ is generated by $I$. So $T$ is generated by $T_n$. 
  
  In particular, this means that $T_n$ generates $T_{n+1}$. We prove that this cannot be the case.
  We consider how ancilla symbols can be used when using $T_n$ to generate $T_{n+1}$. 
  \begin{enumerate}
  \item If an ancilla symbol is initially $1$, then it is $1$ all the time. Any $T_n$ gate whose last input is this ancilla symbol has not effect. 
  So we can assume this ancilla symbol only acts as one of the first $n$ inputs of $T_n$. 
  \item If an ancilla symbol is initially $2$ or $3$, then it is $2$ or $3$ all the time. 
  If it acts as one of the first $n$ inputs of some $T_n$, then that $T_n$ gate has no effect. So we can assume this ancilla symbol only acts as one of the last input of $T_n$. 
  On the other hand, if a $T_n$ gate has last symbol the ancilla symbol, then only that ancilla symbol is affected. This means we do not need this ancilla symbol at all.  
  \item If an ancilla symbol is initially larger than $3$, then any $T_n$ gate acting on this symbol does not have any effect. So we do not need this ancilla symbol. 
  \end{enumerate}
  
  By the above discussion, we can assume that all ancilla symbol are initially $1$'s, and the only use of them is to act as one of the first $n$ inputs of $T_n$. 
  Therefore we only need to show that, using $T_0,\ldots,T_n$ and no ancilla symbols, it is impossible to generate $T_{n+1}$. 
  
  Let $a_1,\ldots,a_n,b$ be the inputs of $T_{n+1}$.
  Consider all $T_i$ ($0\le i\le n$) gates whose last input is not $b$. We can remove these $T_i$ actions, because they do not change the effect of all other $T_i$ actions. 
  So we can assume that the last input of each $T_i$ is $b$. 
  In other words, every $T_i$ acts by choosing a subset of size $i$ among $a_1,\ldots,a_n$ as the first $i$-inputs, and $b$ as the last input. 
  
  Consider $2^n$ different inputs $(a_1,\ldots,a_n,b)$, where the $k$-th input $0\le k<2^n$ has 
  \begin{enumerate}
  \item $a_i=1$ if the $i$-th lowest bit in the binary representation of $k$ is $1$;
  \item $a_i=2$ otherwise;
  \item $b=2$.
  \end{enumerate}
  $T_{n+1}$ changes the value of $b$ for exactly one of the inputs, namely the input $2^n-1$. 
  So we only need to prove that for any gate $A^{n+1}\to A^{n+1}$ built using $T_0,\ldots,T_n$ and no ancilla symbols, there are always an even number of inputs among the $2^k$ ones defined above, on which the value of $b$ is changed.
  
  Each gate $T_i$ ($0\le i\le n$) changes the value of $b$ for an even number of inputs. 
  The symmetric difference of several sets of even size must have even size. 
  So the overall effect is that for an even number of inputs, the value $b$ is changed. 
  \end{proof}
  \begin{rmk}\label{UncountablyMany}
  A natural question to ask is whether there are uncountably many gate classes over a non-binary alphabet. 
  We expect this to be true because this is true for classical gate classes over non-binary alphabets \cite{IM59}. 
  
  A possible approach to proving this is to find a countable collection of gates $F_1,F_2,\ldots$ such that $F_i\not \in \la F_j:j\ne i\ra$ for all $i\ge 1$. If we have such a collection, then for different subsets $I\sse \{1,2,\cdots\}$, the gate classes $F_I=\la F_i:i\in I\ra$ are different. 
  However, we have not been able to find such a collection.
  \end{rmk}
  
  \section{Some generation results}
  In this section, let $A=\{1,\ldots,k\}$ with $k\ge 3$. 
  We prove several generation results which can be seen as generalizations of generation results over the binary alphabet. Nevertheless there are some subtle differences between the binary alphabet and non-binary alphabets. 
  
  Recall Notation \ref{GateDefn} where we defined the reversible gates $\tau_{u,v}$, $\SWAP_{u,v}$, and $w\hyp F$. 
  \subsection{Generating $\ALL$}
  Recall that $\ALL$ is the class of all gates. 
  \begin{thm}\label{AllGenerate}
  Let $S$ be the following set of gates:
  \begin{enumerate}
  \item $\tau_{a,b}$ for all $a\ne b\in A$;
  \item $\tau_{11,12}$.
  \end{enumerate}
  Then $\la S\ra=\ALL$.
  \end{thm}
  \begin{rmk}
  Over the binary alphabet $\{0,1\}$, $\ALL$ is generated by the Toffoli gate $11\hyp\tau_{0,1}$, but not by the CNOT gate $\tau_{10,11}$. Therefore Theorem \ref{AllGenerate} does not hold over the binary alphabet. 
  \end{rmk}
  We prove the theorem in several steps. 
  \begin{lemma}\label{TauABCABD}
  $S$ generates $\tau_{abc,abd}$ for all $a,b,c,d\in A$. 
  \end{lemma}
  \begin{proof}
  \textbf{Step 1.} $S$ generates $\tau_{11,1c}$ for all $c\in A$. 
  Assume we have two input symbols $x$, $y$. The following sequence of operations implements $\tau_{11,1c}$.
  \begin{enumerate}
  \item Apply $\tau_{2,c}$ on $y$. 
  \item Apply $\tau_{11,12}$ on $xy$. 
  \item Apply $\tau_{2,c}$ on $y$.
  \end{enumerate}
  
  \textbf{Step 2.} $S$ generates $\tau_{1b,1c}$ for all $b,c\in A$. This is similar to step 1 and omitted. 
  
  \textbf{Step 3.} $S$ generates $\tau_{ab,ac}$ for all $a,b,c\in A$. Assume we have two input symbols $x$, $y$.
  The following sequence of operations implements $\tau_{ab,ac}$.
  \begin{enumerate}
  \item Apply $\tau_{1,a}$ on $x$. 
  \item Apply $\tau_{1b,1c}$ on $xy$. 
  \item Apply $\tau_{1,a}$ on $x$.
  \end{enumerate}
  
  \textbf{Step 4.} $S$ generates $\tau_{abc,abd}$. Assume we have three input symbols $x$, $y$, $z$. 
  We introduce an ancilla symbol $w$ which is initially $1$. 
  The following sequence of operations implements $\tau_{abc,abd}$ and fixes all ancilla symbols.
  \begin{enumerate}
  \item Apply $\tau_{a1,a2}$ on $xw$. 
  \item Apply $\tau_{b2,b3}$ on $yw$. 
  \item Apply $\tau_{3c,3d}$ on $wz$. 
  \item Apply $\tau_{b2,b3}$ on $yw$. 
  \item Apply $\tau_{a1,a2}$ on $xw$.
  \end{enumerate}
  \end{proof}
  
  \begin{lemma}\label{WControlTrans}
  $S$ generates $w\hyp \tau_{a,b}$ for all strings $w$, and $a,b\in A$.
  \end{lemma}
  \begin{proof}
  If $|w|\le 2$, then $S$ generates $w$ by Lemma \ref{TauABCABD}. 
  In the following, assume $|w|\ge 3$.
  
  Let $|w|=n$. Assume we have $n+1$ inputs $x_1,\ldots,x_n, y$. Introduce $n-1$ ancilla symbols $z_1,\ldots, z_{n-1}$. Initially all $z_i=1$. 
  The following sequence of operations implements $w\hyp\tau_{a,b}$ and fixes all ancilla symbols.
  \begin{enumerate}
  \item Apply $\tau_{w_1w_21,w_1w_22}$ on $x_1x_2z_1$. 
  \item For $i=3,\ldots,n$, apply $\tau_{w_i21,w_i22}$ on $x_iz_{i-2}z_{i-1}$. 
  \item Apply $\tau_{2a,2b}$ on $z_{n-1}y$. 
  \item For $i=n,\ldots,3$, apply $\tau_{w_i21,w_i22}$ on $x_iz_{i-2}z_{i-1}$. 
  \item Apply $\tau_{w_1w_21,w_1w_22}$ on $x_1x_2z_1$. 
  \end{enumerate}
  \end{proof}
  
  \begin{proof}[Proof of Theorem \ref{AllGenerate}]
  We would like to show that for all $l$, all bijections $A^l\to A^l$ are in $\la S\ra$.
  The group of bijections $A^l\to A^l$ is the symmetric group $S_{|A|^l}$, so it is generated by transpositions.
  The composition in the symmetric group is the same as the one in composition rule of reversible gates. So we only need to prove that every transposition is in $\la S\ra$. 
  Moreover, we do not need all transpositions. We only need a set of transpositions $T$ so that the graph whose vertices are $A^k$ and edges are $\{(u,v):\tau_{u,v}\in T\}$ is connected.
  We can take $T$ to be the set of $\tau_{u,v}$ where $u$ and $v$ differ by one position. 
  
  We can assume without loss of generality that $u$ and $v$ differ in the last position.
  Then $\tau_{u,v}$ is in the form $w\hyp\tau_{a,b}$. The theorem follows from Lemma \ref{WControlTrans}. 
  \end{proof}
%   \begin{rmk}
%   In the proposition, the set of all $\tau_{a,b}$ can be replaced with a set of transpositions $\tau_{a,b}$ such that the graph with vertices $A$ and edges $\{(a,b):\tau_{a,b}\in T\}$ is connected. 
%   \end{rmk}
  
  \subsection{Generating $\CONS$}
  Recall that $\CONS$ is the set of conservative gates. 
  \begin{thm}\label{ConsGen}
  Let $S$ be the set of $1\hyp \SWAP_{a,b}$ for all $a\ne b\in A$. 
  Then $\la S\ra=\CONS$.
  \end{thm}
  \begin{rmk}
  Over the binary alphabet $\{0,1\}$, $\CONS$ is generated by the Fredkin gate $1\hyp\SWAP_{0,1}$. So the statement of Theorem \ref{ConsGen} is true over the binary alphabet. However, the proof presented here is for non-binary alphabets. 
  \end{rmk}
  We prove the theorem in several steps. 
  \begin{lemma}
  $S$ generates $c\hyp \SWAP_{a,b}$ for all $a,b,c\in A$. 
  \end{lemma}
  \begin{proof}
  If $c=1$, then the lemma follows from definition of $S$. 
  In the following, assume $c\ne 1$. 
  
  Assume we have three input symbols $x$, $y$, $z$. Introduce an ancilla symbol $w$ which is initially $1$. Introduce an ancilla symbol $u$ which is initially $d$ where $d\ne 1,c$.
  The following sequence of operations implements $c\hyp \SWAP_{a,b}$ and fixes all ancilla symbols. 
  \begin{enumerate}
  \item Apply $\SWAP_{1,d}$ on $ux$. 
  \item Apply $\SWAP_{1,c}$ on $wx$. 
  \item Apply $1\hyp \SWAP_{a,b}$ on $xyz$. 
  \item Apply $\SWAP_{1,c}$ on $wx$. 
  \item Apply $\SWAP_{1,d}$ on $ux$. 
  \end{enumerate}
  \end{proof}
  
  \begin{lemma}\label{WControlSWAP}
  $S$ generates $w\hyp \SWAP_{a,b}$ for all strings $w$ and $a,b\in A$. 
  \end{lemma}
  \begin{proof}
  Let $|w|=n$. Assume that we have $n+2$ inputs $x_1,\ldots,x_n,y_1,y_2$. 
  Introduce $n+1$ ancilla symbols $z_1,\ldots,z_{n+1}$. Initially $z_1=1$ and $z_i=2$ for $i\ge 2$. 
  The following sequence of operations implements $w\hyp \SWAP_{a,b}$ and fixes all ancilla symbols. 
  \begin{enumerate}
  \item For $i=1,\ldots,n$, apply $w_i\hyp\SWAP_{1,2}$ on $x_iz_iz_{i+1}$. 
  \item Apply $1\hyp\SWAP_{a,b}$ on $z_{n+1}y_1y_2$. 
  \item For $i=n,\ldots,1$, apply $w_i\hyp\SWAP_{1,2}$ on $x_iz_iz_{i+1}$. 
  \end{enumerate}
  \end{proof}
  
  \begin{proof}[Proof of Theorem \ref{ConsGen}]
  Clearly $S\sse \CONS$. So we only need to prove that $S$ generates $\CONS$.
  
  Similar to the proof of Theorem \ref{AllGenerate}, we only need to generate all transpositions $\tau_{u,v}$ where $u$ is a permutation of $v$. 
  If we have such $u,v$, then we can find a sequence $u=w_0,w_1,\ldots,w_m=v$ where $w_i$ and $w_{i+1}$ differ by exchanging two positions. 
  So we can assume that $u,v$ differ by exchanging two positions. 
  
  We can assume without loss of generality that the positions at which $u$ and $v$ differ are the last two symbols. Then $\tau_{u,v}$ is in the form $w\hyp\SWAP_{a,b}$. The theorem follows from Lemma \ref{WControlSWAP}. 
  \end{proof}
  
  \subsection{Generating $\CONS_\lm$}
  $\ALL$ and $\CONS$ fall into a more general class of gate classes. 
  \begin{defn}\label{LmCons}
  Let $\lm$ be a partition of $A$, which means $\lm=\{\lm_1,\ldots,\lm_l\}$ where $\lm_1,\ldots,\lm_l$ are disjoint nonempty subsets of $A$ whose union is $A$. 
  Define $\CONS_\lm$ to be the class of \textbf{$\lm$-conservative gates}, which are gates that preserve, for each $\lm_i\in \lm$, the total number of occurrences of elements in $\lm_i$. 
  \end{defn}
  \begin{eg}
  When $\lm=\{A\}$, $\CONS_\lm=\ALL$. 
  When $\lm=\{\{a\}:a\in A\}$, $\CONS_\lm=\CONS$.
  \end{eg}
  The following theorem is a simultaneous generalization of Theorem \ref{AllGenerate} and Theorem \ref{ConsGen}. 
  \begin{thm}\label{LmConsGeneration}
  Let $S$ be the following set gates: 
  \begin{enumerate}
  \item $\tau_{a,b}$ for all $1\le i\le l$ and all $a\ne b\in \lm_i$;
  \item $\tau_{aa,ab}$ for all $1\le i\le l$ with $|\lm_i|>1$ and one pair of $a\ne b\in \lm_i$;
  \item $1\hyp\SWAP_{a,b}$ for all $1\le i<j\le l$ and one pair of $a\in \lm_i$, $b\in \lm_j$.
  \end{enumerate}
  Then $\la S\ra=\CONS_\lm$.
  \end{thm}
  
  We prove the theorem in several steps. 
  \begin{lemma}\label{LmConsGenAll}
  For all $1\le i\le l$, $S$ generates all gates that fix the input if at least one symbol of the input is not in $\lm_i$. 
  \end{lemma}
  \begin{proof}
  This is almost identical to the proof of Theorem \ref{AllGenerate}. Omitted.
  \end{proof}
  \begin{lemma}\label{LmConsGenCons}
  $S$ generates $\CONS$.
  \end{lemma}
  \begin{proof}
  \textbf{Step 1.} Assume in (3) of the definition of $S$ in Theorem \ref{LmConsGeneration}, we choose $a=a_{ij}$, $b=b_{ij}$. Then for all $1\le i<j\le l$, $S$ generates $1\hyp \SWAP_{a_{ij},b_{ij}}$.
  
  \textbf{Step 2.} $S$ generates $1\hyp \SWAP_{a_{ij},c}$ for all $1\le i<j\le l$ and $c\in \lm_j$.
  Assume we have input symbols $x$, $y$, $z$. 
  By Lemma \ref{LmConsGenAll}, $S$ generates $\tau_{b_{ij},c}$. 
  The following sequence of operations implements $1\hyp \SWAP_{a_{ij},c}$.
  \begin{enumerate}
  \item Apply $\tau_{b_{ij},c}$ on $z$. 
  \item Apply $\tau_{b_{ij},c}$ on $y$. 
  \item Apply $1\hyp \SWAP_{a_{ij},b_{ij}}$ on $xyz$. 
  \item Apply $\tau_{b_{ij},c}$ on $z$. 
  \item Apply $\tau_{b_{ij},c}$ on $y$.
  \end{enumerate}
  
  \textbf{Step 3.} $S$ generates $1\hyp \SWAP_{a,b}$ for all $1\le i<j\le l$ and $a\in \lm_i$, $b\in \lm_j$.
  This is similar to step 2 and omitted. 
  
  \textbf{Step 4.} $S$ generates $1\hyp\SWAP_{a,b}$ for all $1\le i\le l$ and $a,b\in \lm_i$.
  If $1\in \lm_i$ then this follows from Lemma \ref{LmConsGenAll}. Assume $1\not \in \lm_i$. By step 3, $S$ generates $\SWAP_{1,a}$. By Lemma \ref{LmConsGenAll}, $S$ generates $a\hyp\SWAP_{a,b}$.
  
  Assume we have input symbols $x$, $y$, $z$. Introduce an ancilla symbol $w$ which is initially $a$. 
  The follows sequence of operations implements $1\hyp\SWAP_{a,b}$ and fixes all ancilla symbols.
  \begin{enumerate}
  \item Apply $\SWAP_{1,a}$ to $wx$. 
  \item Apply $a\hyp\SWAP_{a,b}$ to $xyz$. 
  \item Apply $\SWAP_{1,a}$ to $wx$.
  \end{enumerate}
  
  \textbf{Step 5.} Apply Theorem \ref{ConsGen}.
  \end{proof}
  \begin{lemma}\label{WControlTransLm}
  $S$ generates $\tau_{wa,wb}$ where $w$ is a string and $a,b\in \lm_i$ for some $1\le i\le l$.
  \end{lemma}
  \begin{proof}
  Assume $|w|=n$. 
  Assume we have $n+1$ inputs $x_1,\ldots,x_n,y$.
  Introduce two ancilla symbols $z_1,z_2$. Initially, $z_1=a$, $z_2=b$.
  $S$ generates $w\hyp\SWAP_{a,b}$ by Lemma \ref{LmConsGenCons}.
  $S$ generates $\tau_{aa,ab}$ by Lemma \ref{LmConsGenAll}.
  
  The following sequence of operations implements $\tau_{wa,wb}$ and fixes all ancilla symbols.
  \begin{enumerate}
  \item Apply $w\hyp\SWAP_{a,b}$ on $x_1\ldots x_nz_1z_2$. 
  \item Apply $\tau_{aa,ab}$ on $z_2y$. 
  \item Apply $w\hyp\SWAP_{a,b}$ on $x_1\ldots x_nz_1z_2$. 
  \end{enumerate}
  \end{proof}
  \begin{proof}[Proof of Theorem \ref{LmConsGeneration}]
  Clearly $S\sse \CONS_\lm$. We only need to prove that $S$ generates $\CONS_\lm$.

  Similar to the proof of Theorem \ref{AllGenerate}, we only need to prove that $S$ generates transpositions $\tau_{u,v}$ where $u,v\in A^m$ for some $m$, and for all $1\le i\le l$, the number of occurrences in $u$ of elements in $\lm_i$ is the same as that of $v$.
  
  Assume $u,v\in A^m$ has the above property.
  We can find $w\in A^m$ such that $w$ is a permutation of $v$, and for all $1\le j\le m$, $w_j$ and $u_j$ are both in $\lm_i$ for some $1\le i\le l$. 
  $S$ generates $\tau_{w,v}$ by Lemma \ref{LmConsGenCons}. So we can replace $v$ with $w$ and assume that for all $1\le j\le m$, $u_j$ and $v_j$ are both in $\lm_i$ for some $1\le i\le l$.
  
  Assume we have such $u,v$. Then we can find a sequence $u=w_0,w_1,\ldots,w_p=v$ such that $w_i$ and $w_{i+1}$ differ by one position $j$, and $w_{i,j}$ and $w_{i+1,j}$ are both in $\lm_h$ for some $1\le h\le l$. So we can assume that $u$ and $v$ satisfies this.
  
  So we only need to generate $\tau_{wa,wb}$ where $w$ is some string, and $a,b\in \lm_i$ for some $1\le i\le l$. This follows from Lemma \ref{WControlTransLm}.
  \end{proof}
  
  \subsection{Single gate generation}
  The following proposition points out the reason why finitely-generated gate classes are usually generated by a single gate. 
  \begin{prop}\label{SingleGeneration}
  Let $S$ be a finite set of reversible gates such that each gate in $S$ has a fixed point.
  Then $\la S\ra$ is generated by a single gate.
  \end{prop}
  \begin{proof}
  Let $S=\{F_1,\ldots,F_n\}$. 
  Consider the gate $F=F_1\ot \cdots \ot F_n$. By the tensor product rule of generation, $F\in \la S\ra$. So we only need to prove that $F$ generates $F_i$ for all $i$. 
  
  Fix an $i$. We use the ancilla rule of generation. Because every $F_j$ has a fixed point, for $j\ne i$, we can take the input gates for $F$ that correspond to $F_j$ to be a fixed point of $F_j$. Then $F$ fixes all inputs symbols except for the inputs corresponding to $F_i$. 
  \end{proof}
  \begin{coro}
  $\CONS_\lm$ is generated by a single gate for all $\lm$.
  \end{coro}
  \begin{proof}
  Each gate in the set $S$ in Theorem \ref{LmConsGeneration} has fixed points. 
  \end{proof}
  
  \section{Base change}
  \begin{prop}\label{SkAction}
  Let $L_k$ be the lattice of gate classes when $|A|=k$. 
  Then $L_k$ is naturally equipped with a nontrivial $S_k$-action. 
  \end{prop}
  \begin{proof}
  Let $\sm\in S_k$ be a permutation of $A$.
  For a gate $F:A^l\to A^l$, define reversible gate $F^\sm:A^l\to A^l$ that maps $(a_1,\ldots,a_l)\in A^l$ to $\sm(F(\sm^{-1}(a_1),\ldots,\sm^{-1}(a_l)))$, where the outer $\sm$ means applying $\sm$ on every input symbol.
  For every gate class $S$, define $S^\sm=\{F^\sm:F\in S\}$.
  It is easy to verify that $S^\sm$ is also a gate class. 
  Then $\sm:L_k\to L_k$ gives the desired $S_k$-action.
  \end{proof}
  \begin{rmk}
  This generalizes the notion of dual gate and dual gate class when $|A|=2$ defined in \cite{AGS15}.
  %It is known that every gate class is dual-closed when $|A|=2$ (\cite{AGS15}, Corollary 6). This fails badly when $|A|\ge 3$. 
  \end{rmk}
  
  \begin{thm}\label{PosetEmbedding}
  Let $L_k$ be the lattice of gate classes when $|A|=k$.
  Then there is a poset embedding $P:L_k\to L_{k+1}$ for $k\ge 2$.
  \end{thm}
  
  \begin{proof}
  Let $F$ be a reversible gate over alphabet $A=\{1,\ldots,k\}$. 
  Define $PF$ to be the reversible gate over alphabet $A^\p=\{1,\ldots,k+1\}$ that fixes the input when at least one of the input symbols is $k+1$, and acts as $F$ otherwise. 
  
  Let $S\in L_k$ be a gate class. Define $PS=\la PF:F\in S\ra$ to be a gate class over $A^\p$. This defines a map between sets $P:L_k\to L_{k+1}$. $P$ is clearly order-preserving, so $P$ is a map between posets. 
  
  So we only need to prove that $P$ is an embedding, i.e. $P$ is an injection of sets.
  This means for every two different gate classes $S,T\in L_k$, we have $PS\ne PT$.
  $PT=\la PF:F\in T\ra$. So we only need to prove that there exists some $F\in T$ such that $PF\not \in PS$. 
  Actually, we prove that for any reversible gate $F$ over $A$, if $PF\in PS$, then $F\in S$.
  
  $PF\in PS$, so $PF$ can be generated by $\{PG:G\in S\}$.
  Consider how ancilla symbols are used to perform this generation.
  
  If an ancilla symbol is initially $k+1$, then it remains $k+1$ all the time. If this ancilla symbol is used as an input of some gate $PG$, then this $PG$ action has no effect. So we do not need this ancilla symbol . 
  
  We can perform the generation with all ancilla symbols initially in $A$.
  This gives a way to generate $F$ using $S$. 
  \end{proof}
  \begin{rmk}\label{FindLatticeEmbedding}
  We do not see the reason for $P$ to be either a join-lattice embedding or a meet-lattice embedding. 
  It is an interesting question whether we can find some map $Q:L_{k}\to L_{k+1}$ that is a lattice embedding (or even a complete lattice embedding). 
%   It is also possible that we can find some other map $Q:L_k\to L_{k+1}$ which is a lattice embedding. 
%   For any gate class $T\in L_{k+1}$, the map $Q_T:L_k\to L_{k+1}$ which sends $S$ to $\la PS,T\ra$ is a poset homomorphism. We do not know whether any such $Q_T$ gives a lattice embedding. 
  \end{rmk}
  
  \section{Gate classes containing $\CONS_{k-1,1}$}
  Let $A=\{1,\ldots,k\}$ where $k\ge 3$. 
  By Proposition \ref{SkAction}, there is a natural action of $S_k$ on $L_k$.
  Therefore if we have two partitions $\lm=\{\lm_1,\ldots,\lm_l\}$ and $\mu=\{\mu_1,\ldots,\mu_l\}$ of $A$ such that $|\lm_i|=|\mu_i|$ for all $i$, then the gate classes $\CONS_\lm$ and $\CONS_\mu$ are ``isomorphic'' in a certain sense.
  \begin{defn}
  Let $\tau$ be a partition of $k$, i.e. $\tau=\{\tau_1,\ldots,\tau_l\}$ where $\tau_i\ge 1$ and $\sum \tau_i=k$. Define $\CONS_\tau$ to be any $\CONS_\lm$ where $\lm=\{\lm_1,\ldots,\lm_l\}$ is a partition of $A$ and $|\lm_i|=\tau_i$ for all $i$. 
  $\CONS_\lm$ is well-defined up to $S_k$-action. 
  \end{defn}
  In this section we study the properties of $\CONS_{k-1,1}$.
  We only need to study $\CONS_\lm$ where $\lm=\{\{1,\ldots,k-1\},\{k\}\}$. 
  
  We adapt the methods in \cite{AGS15} and prove the following theorem.
  \begin{thm}\label{CONSInteger}
  The lattice of gate classes containing $\CONS_{k-1,1}$ is anti-isomorphic to the lattice $(\bN,|)$ of non-negative integers, where $n\le m$ in the lattice order relation if and only if $n|m$.
  The gate class corresponding to $0$ is $\CONS_{k-1,1}$, and the gate class corresponding to $m\ge 1$ is $\la \CONS_{k-1,1},\CC_m\ra$, where $\CC_m=1\hyp \tau_{1^m,k^m}$. 
  \end{thm}
  The Hasse diagram of the lattice of gate classes containing $\CONS_{k-1,1}$ is shown in Figure \ref{FigCONSInteger}.
  \begin{rmk}
  Theorem \ref{CONSInteger} fails when $k=2$ because there are parity-flipping gates in that case.
  In particular, Lemma \ref{NoModShifter} fails when $k=2$.
  \end{rmk}
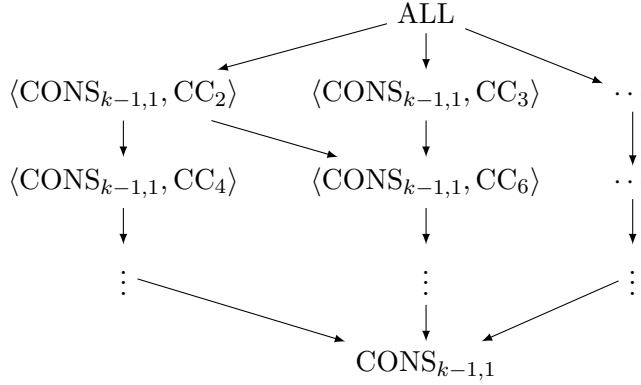
\begin{figure}%[h]
\begin{center}
\begin{tikzpicture}[>=latex]
\matrix[row sep=0.5cm,column sep=0.7cm] {
& \node (ALL) {$\ALL$}; & \\
\node (MOD2) {$\la \CONS_{k-1,1},\CC_2\ra$}; & \node (MOD3) {$\la \CONS_{k-1,1},\CC_3\ra$}; & \node (MOD5) {$\cdots$}; \\
\node (MOD4) {$\la \CONS_{k-1,1},\CC_4\ra$}; & \node (MOD6) {$\la \CONS_{k-1,1},\CC_6\ra$};& \node (MOD10) {$\cdots$}; \\
\node (MOD8) {$\vdots$}; & \node (MOD18) {$\vdots$};& \node (MOD20) {$\vdots$}; \\
& \node (CONS) {$\CONS_{k-1,1}$}; &\\
};
% EDGES
\path[draw,->] 
(ALL) edge (MOD2)
(ALL) edge (MOD3)
(ALL) edge (MOD5)
(MOD2) edge (MOD4)
(MOD2) edge (MOD6)
(MOD3) edge (MOD6)
(MOD5) edge (MOD10)
(MOD4) edge (MOD8)
(MOD6) edge (MOD18)
(MOD10) edge (MOD20)
(MOD8) edge (CONS)
(MOD18) edge (CONS)
(MOD20) edge (CONS)
;
\end{tikzpicture}
\par
\end{center}
\caption{The lattice of gate classes containing $\CONS_{k-1,1}$}
\label{FigCONSInteger}
\end{figure}

  The theorem is proved in several steps.
  \begin{notn}
  For any finite string $w$ of elements in $A$, define $c_k(w)$ to be the number of $k$'s in the string.
  \end{notn}
  \begin{defn}
  A gate $F:A^n\to A^n$ is called \textbf{mod-$m$-respecting} if for every $w\in A^n$, $c_k(F(w))=c_k(w)+j \pmod m$ for some $0\le j\le m-1$. When $j=0$, $F$ is called \textbf{mod-$m$-preserving}. 
  \end{defn}
  \begin{lemma}[Adaptation of \cite{AGS15}, Theorem 12]\label{NoModShifter}
  If $F$ is mod-$m$-respecting, then it is mod-$m$-preserving.
  \end{lemma}
  \begin{proof}
  Let $z$ be a variable. 
  $c_k(F(w))=c_k(w)+j\pmod k$ means $z^{c_k(F(w))}=z^{c_k(w)+j} \pmod{(z^m-1)}$. 
  We sum over all $w\in A^n$ and get
  \begin{align*}
  \sum_{w\in A^n} z^{c_k(F(w))}=z^j \sum_{w\in A^n} z^{c_k(w)} \pmod {(z^m-1)}.
  \end{align*}
  On the other hand, we have 
  \begin{align*}
  \sum_{w\in A^n} z^{c_k(F(w))}=\sum_{w\in A^n} z^{c_k(w)}=\sum_{0\le i\le n} \binom ni(k-1)^{n-i}z^i=(z+k-1)^n
  \end{align*}
  From the above two equalities, we see $(z+k-1)^n(z^j-1)=0\pmod {(z^m-1)}$. 
  $k\ge 3$ so $z+k-1$ is coprime with $z^m-1$. Therefore $z^j-1=0\pmod{(z^m-1)}$.
  This is true only when $j=0$. 
  \end{proof}
  \begin{defn}\label{ModPreserverM}
  For a gate $F$, define $m(F)$ to be the largest $m$ such that $F$ is a mod-$m$-preserving. If there is no such largest $m$, define $m(F)=0$.
  \end{defn}
  \begin{lemma}[Adaptation of \cite{AGS15}, Proposition 2]\label{TensorModPreserver}
  For two gates $F$, $G$, we have $m(F\ot G)=\gcd(m(F),m(G))$.
  \end{lemma}
  \begin{proof}
  Clearly $\gcd(m(F),m(G))|m(F\ot G)$. 
  Assume $m(F\ot G)=m^\p\ne \gcd(m(F),m(G))$. 
  Then either $m^\p\nmid m(F)$ or $m^\p\nmid m(G)$. 
  Assume without loss of generality $m^\p\nmid m(F)$. 
  If $F$ is mod-$m^\p$-preserving, then it is mod-$\lcm(m^\p,m(F))$-preserving, and there is contradiction.
  So $F$ is not mod-$m^\p$-preserving. 
  By Lemma \ref{NoModShifter}, $F$ is not mod-$m^\p$-respecting.
  So we can find two inputs $u$, $v$ of $F$ such that $c_k(Fu)-c_k(u)\ne c_k(Fv)-c_k(v)\pmod {m^\p}$.
  Then for any input $w$ of $G$, $c_k((F\ot G)uw)-c_k(uw)\ne c_k((F\ot G)vw)-c_k(vw)\pmod {m^\p}$, which means that $F\ot G$ is not mod-$m^\p$-preserving. Contradiction. 
  So $m(F\ot G)=\gcd(m(F),m(G))$.
  \end{proof}
  \begin{prop}[Adaptation of \cite{AGS15}, Theorem 27]\label{ModPreserverFromCCm}
  For any $m\ge 1$, $\CONS_{k-1,1}$ together with $\CC_m$ generates the class of all mod-$m$-preserving gates.
  \end{prop}
  \begin{proof}
  Similar to the proof of Theorem \ref{AllGenerate}, we only need to show that $\CONS_{k-1,1}+\CC_m$ generates transpositions $\tau_{u,v}$ where $u,v\in A^n$ for some $n$ and $c_k(u)=c_k(v)\pmod m$.
  
  If $c_k(u)=c_k(v)$, then $\tau_{u,v}\in \CONS_{k-1,1}$. 
  So we can assume $c_k(u)\ne c_k(v)$. Assume without loss of generality $c_k(u)<c_k(v)$. 
  Then we can find a sequence $u=w_0,w_1,\ldots,w_p=v$ such that $c_k(w_i)=c_k(w_{i+1})-m$.
  So we can assume $c_k(v)-c_k(u)=m$.
  
  Choose a string $w\in A^n$ such that $w$ is a permutation of $v$, and the set of positions $j$ where $w_j=k$ contains the set of positions $j$ where $u_j=k$.
  $\tau_{v,w}\in \CONS_{k-1,1}$, so we can replace $v$ with $w$.
  
  Therefore we can assume the set of positions $j$ where $v_j=k$ contains the set of positions $j$ where $u_j=k$.
  Assume without loss of generality that for $1\le j\le m$, $v_j=k$ but $u_j\ne k$.
  Let $w$ be the string whose first $m$ symbols are $k$, and the last $n-m$ symbols are the same as $u$. 
  Then $\tau_{w,v}\in \CONS_{k-1,1}$. So we can replace $v$ with $w$ and assume that the last $n-m$ symbols of $u$ are the same as that of $v$. 
  
  Let $w$ be the string whose whose first $m$ symbols are $1$, and the last $n-m$ symbols are the same as $u$. 
  Then $\tau_{w,u}\in \CONS_{k-1,1}$. So we can replace $u$ with $w$ and assume that the first $m$ symbols of $u$ are all $1$.
  
  Now we generate $\tau_{u,v}$.
  Let $x_1,\ldots,x_n$ be inputs.
  Introduce an ancilla symbol $z$, initialized to be $2$. 
  The following sequence of operations implements $\tau_{u,v}$ and fixes all ancilla symbols.
  \begin{enumerate}
  \item Apply $u_{m+1}\ldots u_n\hyp \tau_{1,2}$ on $x_{m+1}\ldots x_nz$.
  \item Apply $\CC_m$ on $zx_1\ldots x_m$. 
  \item Apply $u_{m+1}\ldots u_n\hyp \tau_{1,2}$ on $x_{m+1}\ldots x_nz$.
  \end{enumerate}
  \end{proof}
  \begin{lemma}[\cite{AGS15}, Proposition 40]\label{TensorW}
  Let $F:A^n\to A^n$ be a gate with $m(F)\ge 1$. Then there exists $t\ge 1$ and string $w\in A^{nt}$ such that $c_k(F^{\ot t}(w))-c_k(w)=m$.
  \end{lemma}
  \begin{proof}
  Let $p=\gcd_{w\in A^n} (c_k(F(w))-c_k(w))$. By definition $F$ is mod-$p$-respecting. So $F$ is mod-$p$-preserving, i.e. $p|m(F)$, by Lemma \ref{NoModShifter}.
  There are both positive and negative elements in $\{c_k(F(w))-c_k(w):w\in A^n\}$. So we can find $w_1,\ldots,w_t\in A^n$ such that $\sum_{1\le i\le t} (c_k(F(w_i))-c_k(w_i))=m$. Then $w_1\ldots w_t\in A^{nt}$ is the desired string.
  \end{proof}
  \begin{prop}[Adaptation of \cite{AGS15}, Theorem 41]\label{CCmFromModPreserver}
  For any $m\ge 1$, $\CONS_{k-1,1}$ together with any gate $F$ with $m(F)=m$ generates $\CC_m$.
  \end{prop}
  \begin{proof}
  Let $G$ be $F\ot F^{-1}$ followed by swapping the inputs of two $F$'s. Then $G^2=\id$. 
  By Lemma \ref{TensorModPreserver}, $m(G)=m$. 
  By Lemma \ref{TensorW}, there exists $t\ge 1$ and two strings $u$ and $v$ (whose lengths are the same as the number of inputs of $G^{\ot t}$) such that $v=G^{\ot t}(u)$ and $c_k(v)=c_k(u)+m$. 
  We also have $G^{\ot t}(v)=u$ because $G$ is an involution.
  Let $|u|=n$. 
  
  Assume without loss of generality that the first $c_k(v)$ symbols of $v$ are $k$.
  Because we have $\CONS_{k-1,1}$, we have a gate $H:A^n\to A^n$ satisfying the following:
  \begin{enumerate}
  \item $H(u)=H^{-1}(u)=k^{c_k(u)}1^{n-c_k(u)}$;
  \item $H(v)=H^{-1}(v)=k^{c_k(v)}1^{n-c_k(v)}$.
  \end{enumerate}
  
  Let $w$ be any string of length $n$ of form $c_1,F(c_1),c_2,F(c_2),\ldots,c_t,F(c_t)$ where $c_i$ are strings of length equal to the number of inputs of $F$. 
  Then $G^{\ot t}(w)=w$. 
  
  Let $R=H\circ G^{\ot t}\circ H^{-1}$.
  Let $w^\p=H(w)$, $u^\p=H(u)$ and $v^\p=H(v)$. 
  Direct calculation shows that  $R(w^\p)=w^\p$, $R(u^\p)=v^\p$ and $R(v^\p)=u^\p$.
  
  Now we use $R$ to generate $\CC_m$.
  Let $c,x_{c_k(u)+1},\ldots,x_{c_k(v)}$ be inputs. Introduce ancilla symbols: 
  \begin{enumerate}
  \item $x_1,\ldots,x_{c_k(u)}$, initialized to $1$;
  \item $x_{c_k(v)+1},\ldots,x_n$, initialized to $k$;
  \item $y_1,\ldots,y_n$, initialized to $w^\p$;
  \item $z_1,z_2$ where $z_1$ initialized to $1$, $z_2$ initialized to $2$.
  \end{enumerate}
  The following sequence of operations implements $\CC_m$ and fixes all ancilla symbols.
  \begin{enumerate}
  \item Swap $z_1$ and $z_2$ if $c=1$ and $x_{c_k(u)+1}\ldots x_{c_k(v)}=1^m$ or $k^m$.
  \item Swap $x_1\ldots x_n$ with $y_1\ldots y_n$ if $z_1=1$. 
  \item Apply $R$ on $x_1\ldots x_n$. 
  \item Swap $x_1\ldots x_n$ with $y_1\ldots y_n$ if $z_1=1$. 
  \item Swap $z_1$ and $z_2$ if $c=1$ and $x_{c_k(u)+1}\ldots x_{c_k(v)}=1^m$ or $k^m$.
  \end{enumerate}
  
  \end{proof}
  \begin{proof}[Proof of Theorem \ref{CONSInteger}]
  Let $S$ be a gate class containing $\CONS_{k-1,1}$. Assume $S\ne \CONS_{k-1,1}$.
  Define $m=\gcd_{F\in S} m(F)$.
  By Proposition \ref{ModPreserverFromCCm}, $S$ is contained in $\la \CONS_{k-1,1}, \CC_m\ra$.
  
  We can find a finite number of gates $F_1,\ldots,F_n\in S$ such that $m=\gcd_{1\le i\le n} m(F_i)$.
  By Lemma \ref{TensorModPreserver}, $m(F_1\ot \cdots \ot F_n)=m$.
  By Proposition \ref{CCmFromModPreserver}, $\CONS_{k-1,1}$ with $F_1\ot \cdots \ot F_n$ generates $\la \CONS_{k-1,1},\CC_m\ra$.
  So $S=\la \CONS_{k-1,1},\CC_m\ra$. . 
  \end{proof}
  
  \section{Further directions}
  There are many possible directions for further research on reversible gate classes over non-binary alphabets. 
  As discussed in Remark \ref{UncountablyMany}, it remains open whether there are uncountably many reversible gate classes over a non-binary alphabet. 
  As discussed in Remark \ref{FindLatticeEmbedding}, we do not know whether there is a lattice embedding from $L_k$ to $L_{k+1}$. 
  
  %Je\v r\'abek \cite{Jer14} developed a reversible clone-coclone duality, which is reviewed in Appendix \ref{CloneCoclone}. We would like to adjust his definitions to have an explicit characterization of the coclones. An attempt is recorded in Appendix \ref{ExplicitCoclone}. 
  
  We have seen from Theorem \ref{NonFinGenClass} that reversible gate classes over non-binary alphabets can be much more complicated than those over the binary alphabet. Therefore, it seems very hard to give a complete classification of reversible gate classes over non-binary alphabets. However, when restricted to certain kinds of gate classes, classification can be done. Theorem \ref{CONSInteger} is one result in this spirit. We expect there to be more such classification to reveal more structure of the huge lattice $L_k$. For example, it would be interesting to have a classification of gate classes containing $\CONS_{k-2,1,1}$, or of gate classes containing all one-input gates.

  \bibliographystyle{alpha}
  \bibliography{ref}

\end{document}